\documentclass{amsart}

\newtheorem{theorem}{Theorem}[section]
\newtheorem{lemma}[theorem]{Lemma}

\newtheorem{proposition}[theorem]{Proposition}

\theoremstyle{definition}

\theoremstyle{remark}
\newtheorem{remark}[theorem]{Remark}

\numberwithin{equation}{section}

\begin{document}

\title[A New Efficient Asymmetric Cryptosystem Based on the Square Root Problem] {A New Efficient Asymmetric Cryptosystem Based on the Square Root Problem}

\author{M.R.K. Ariffin}
\address{Al-Kindi Cryptography Research Laboratory, Institute for Mathematical Research,
Universiti Putra Malaysia, 43400 UPM,Serdang, Selangor, MALAYSIA}
\curraddr{Department of Mathematics, Faculty of Science, Universiti
Putra Malaysia, 43400 UPM, Serdang, Selangor, MALAYSIA}
\email{rezal@putra.upm.edu.my}
\thanks{The research was supported by the Fundamental Research Grant Scheme
$\#5523934$ and Prototype Research Grant Scheme $\#5528100$ Ministry
of Higher Education, MALAYSIA.}

\author{M.A.Asbullah}
\address{Al-Kindi Cryptography Research Laboratory, Institute for Mathematical Research,
Universiti Putra Malaysia, 43400 UPM,Serdang, Selangor, MALAYSIA}

\curraddr{Centre of Foundation Studies for Agricultural Science,
Universiti Putra Malaysia, 43400 UPM, Serdang, Selangor, MALAYSIA}
\email{ma$\_$asyraf@putra.upm.edu.my}

\author{N.A. Abu}
\address{Al-Kindi Cryptography Research Laboratory, Institute for Mathematical Research,
Universiti Putra Malaysia, 43400 UPM,Serdang, Selangor, MALAYSIA}
\curraddr{Department of Computer Systems and Communication,
Faculty of Information and Communication Technology, Universiti
Teknikal Malaysia Melaka, 76109 Durian Tunggal, Melaka, MALAYSIA}
\email{nura@utem.edu.my}







\begin{abstract}
The square root modulo problem is a known primitive in designing an
asymmetric cryptosystem. It was first attempted by Rabin. Decryption
failure of the Rabin cryptosystem caused by the 4-to-1 decryption
output is overcome efficiently in this work. The proposed scheme
(known as the $AA_{\beta}$-cryptosystem) has its encryption speed
having a complexity order faster than the Diffie-Hellman Key
Exchange, El-Gammal, RSA and ECC. It can also transmit a larger data
set securely when compared to existing asymmetric schemes. It has a
simple mathematical structure. Thus, it would have low computational
requirements and would enable communication devices with low
computing power to deploy secure communication procedures
efficiently.
\end{abstract}

\maketitle



\section{Introduction}
The Rabin cryptosystem that utilizes the square root modulo
problem, is said to be an optimal implementation of RSA with the
encryption exponent $e=2$ \cite{rabin}. However, the situation of
a 4-to-1 mapping during decryption has deterred it from being
utilized. Mechanisms to ensure its possible implementation have
been proposed, however the solutions either still have a
possibility of decryption failure or the performance against the
RSA is inadequate. As a consequence other underlying cryptographic
primitives have taken centre stage. The discrete log problem (DLP)
and the elliptic curve discrete log problem (ECDLP) has been the
source of security for cryptographic schemes such as the Diffie
Hellman key exchange (DHKE) procedure, El-Gamal cryptosystem and
elliptic curve cryptosystem (ECC) respectively \cite{diffie},
\cite{koblitz}. As for the world renowned RSA cryptosystem, the
inability to find the $e$-th root of the ciphertext C modulo N
from the congruence relation $C\equiv M^e (\textrm{mod }N)$
coupled with the inability to factor $N=pq$ for large primes $p$
and $q$ is its fundamental source of security \cite{rsa}. It has
been suggested that the ECC is able to produce the same level of
security as the RSA with shorter key length. Thus, ECC should be
the preferred asymmetric cryptosystem when compared to RSA
\cite{vanstone}. Hence, the notion ``cryptographic efficiency" is
conjured. That is, to produce an asymmetric cryptographic scheme
that could produce security equivalent to a certain key length of
the traditional RSA but utilizing shorter keys. However, in
certain situations where a large block needs to be encrypted, RSA
is the better option than ECC because ECC would need more
computational effort to undergo such a task \cite{scott}. Thus,
adding another characteristic toward the notion of ``cryptographic
efficiency" which is it must be less ``computational intensive"
and be able to transmit large blocks of data (when needed). In
1998 the cryptographic scheme known as NTRU was proposed with
better "cryptographic efficiency" relative to RSA and ECC
\cite{hoffstein2} \cite{hermans} \cite{hoffstein3}. NTRU has a
complexity order of $O(n^{2})$ for both encryption and decryption
as compared to DHKE, EL-Gammal, RSA and ECC (all have a complexity
order of $O(n^{3})$). As such, in order to design a
state-of-the-art public key mechanism, the following are
characteristics that must be ``ideally" achieved (apart from other
well known security issues):

\begin{enumerate}
    \item Shorter key length. If possible shorter than ECC 160-bits.
    \item Speed. To have speed of complexity order $O(n^2)$ for both encryption and decryption.
    \item Able to increase data set to be transmitted asymmetrically. That is, not to be restricted in size because of the mathematical structure.
    \item Simple mathematical structure for easy implementation.
\end{enumerate}

In this paper, we attempt to efficiently enhance an asymmetric
cryptosystem based on the square root problem as its cryptogrpahic
primitive. That is, we will efficiently redesign Rabin's cryptosytem
that has decryption failure due to a 4-to-1 mapping. We will show
that in our design for encryption, it does not involve ``expensive"
mathematical operation. Only basic multiplication is required
without division or modulo operation.\\

The layout of this paper is as follows. The Rabin cryptosystem
will be discussed in Section 2. Previous designs to overcome the
decryption failure of the Rabin cryptosystem will also be
presented here. The mechanism of the $AA_{\beta}$-cryptosystem
will be detailed in Section 3. In Section 4, the authors detail
the decryption process and provide a proof of correctness. An
example will also be presented. Continuing in Section 5, we will
discuss a congruence attack, a Coppersmith type attack and a
Euclidean division attack. An analysis of lattice based attack
will be given in Section 6. Section 7 will be about the underlying
security principles of the $AA_{\beta}$ scheme. A table of
comparison between the $AA_{\beta}$ scheme against RSA,ECC and
NTRU is given in Section 8. Finally, we shall conclude in Section
9.

\section{The Rabin Cryptosystem}

Let us begin by stating that the communication process is between
A (Along) and B (Busu), where Busu is sending information to Along
after encrypting the plaintext with Along's public key.\\

$\bullet$ \textbf{Key Generation by Along}\\
\newline
\indent INPUT: Generate two random $n$-bit prime numbers $p$ and
$q$.\\
\indent OUTPUT: The public key $N=pq$ and the private key pair
$(p,q)$.

\begin{remark}
To simplify computation one may choose $p\equiv q\equiv3(\textrm{mod
}4)$.
\end{remark}

$\bullet$ \textbf{Encryption by Busu}\\
\newline
\indent INPUT: The public key $N$ and the message $M$ where $0\leq M\leq N-1$. \\
\indent OUTPUT: The ciphertext $C=M^{2}(\textrm{mod }N)$.\\

$\bullet$ \textbf{Decryption by Along}\\
\newline
\indent INPUT: The private key pair $(p,q)$ and the ciphertext $C$. \\
\indent OUTPUT: The plaintext $M$.

\begin{remark}
Computing the square roots of $C$ modulo $N$ using the private keys
$(p,q)$, would result in 4 square roots of $C$ modulo $N$. Thus, the
``infamous" decryption failure scenario.
\end{remark}

\subsection{Redundancy Schemes for Unique Decryption}
In order to overcome the decryption failure, it is necessary to
have a scheme that could provide the plaintext upon decryption
without having to guess. We provide here a brief description of 3
existing solution techniques.\\

\begin{enumerate}
  \item \textbf{Redundancy in the message \cite{menezes}}. This scheme has a
  probability decryption failure of approximately $\frac{1}{2^{l-1}}$
  where $l$ is the least significant binary string of the message.\\
  \item \textbf{Extra bits \cite{galbraith}}. One will send 2 extra bits of
  information to specify the square root. The encryption process requires the
  computation of the Jacobi symbol. This results in a computational overhead
  which is much more than just computing a single square modulo $N$. \\
  \item \textbf{Williams technique \cite{williams}}. The encryption process
  requires the encrypter   to compute a Jacobi symbol. Hence,
  losing the performance advantage of Rabin over RSA (as in point no.2).\\
\end{enumerate}

In the next section, we will present an efficient enhancement of
Rabin's cryptosystem that does not inherit the above properties.

\section{The $AA_{\beta}$ Public Key Cryptosystem}

$\bullet$ \textbf{Key Generation by Along}\\
\newline
\indent INPUT: The size $n$ of the prime numbers.\\
\indent OUTPUT: A public key tuple $(n,e_{A1},e_{A2})$ and a
private key pair $(pq, d)$.\\
\begin{enumerate}
    \item Generate two random and distinct $n$-bit strong primes $p$ and
    $q$ satisfying
$$
\left\{%
\begin{array}{ll}
    p\equiv3(\textrm{mod }4), & \hbox{$2^{n}<p<2^{n+1}$,} \\
    q\equiv3(\textrm{mod }4), & \hbox{$2^{n}<q<2^{n+1}$.} \\
\end{array}%
\right.
$$
    \item Choose random $d$ such that $d>(p^{2}q)^{\frac{4}{9}}$.
    \item Choose random integer $e$ such that
    $ed\equiv 1(\textrm{mod } pq)$ and add multiples of $pq$ until
    $2^{3n+4}<e<2^{3n+6}$ (if necessary).
    \item Set $e_{A1}=p^{2}q$. We have $2^{3n}<e_{A1}<2^{3n+3}.$
    \item Set $e_{A2}=e$.
    \item Return the public key tuple $(n,e_{A1},e_{A2})$ and a
private key pair $(pq,d)$.\\
\end{enumerate}

\noindent We also have the fact that $2^{2n}<pq<2^{2n+2}.$\\

$\bullet$ \textbf{Encryption by Busu}\\
\newline
\indent INPUT: The public key tuple $(n,e_{A1},e_{A2})$ and the message \textbf{M}.\\
\indent OUTPUT: The ciphertext $C$.\\
\begin{enumerate}
    \item Represent the message \textbf{M} as a $4n$-bit integer $m$ within the interval $(2^{4n-1},2^{4n})$ with
    $m=m_{1}\cdot2^{n}+m_{2}$ where $m_{1}$ is a $3n+1$-bit integer within the interval $(2^{3n},2^{3n+1})$
    and $m_{2}$ is a $n-1$-bit integer within the interval $(2^{n-2},2^{n-1})$.
    \item Choose a random $n$-bit integer $k_1$ within the interval $(2^{n-1},2^{n})$ and compute
    $U=m_{1}\cdot2^{n}+k_{1}$. We have $2^{4n}<U<2^{4n+1}$.
    \item Choose a random $n$-bit integer $k_2$ within the interval $(2^{n-1},2^{n})$ and compute and compute
    $V=m_{2}\cdot2^{n}+k_{2}$. We have $2^{2n-2}<V<2^{2n-1}$.
    \item Compute $C=Ue_{A1} + V^{2}e_{A2}$.
    \item Send ciphertext $C$ to Along.
\end{enumerate}

\section{Decryption}

\begin{proposition}
Decryption by Along is conducted in the following steps:\\
\newline
\indent INPUT: The private key $(pq, d)$ and the ciphertext $C$.\\
\indent OUTPUT: The plaintext \textbf{M}.\\
\begin{enumerate}
    \item Compute $W\equiv Cd(\textrm{mod }pq).$
    \item Compute $M_{1}\equiv q^{-1}(\textrm{mod }p)$ and $M_{2}\equiv p^{-1}(\textrm{mod
    }q).$
    \item Compute $$x_{p}\equiv W^{\frac{p+1}{4}}(\textrm{mod }p),
    x_{q}\equiv W^{\frac{q+1}{4}}(\textrm{mod }q).$$
    \item Compute $$V_{1}\equiv x_{p}M_{1}q+x_{q}M_{2}p \textrm{ }(\textrm{mod
    }pq),$$ $$V_{2}\equiv x_{p}M_{1}q-x_{q}M_{2}p\textrm{ }(\textrm{mod
    }pq),$$ $$V_{3}\equiv -x_{p}M_{1}q+x_{q}M_{2}p\textrm{ }(\textrm{mod
    }pq),$$ $$V_{4}\equiv -x_{p}M_{1}q-x_{q}M_{2}p\textrm{ }(\textrm{mod
    }pq).$$
    \item For $i=1,2,3,4$ compute
    $U_{i}=\frac{C-V_{i}^{2}e_{A2}}{e_{A1}}.$
    \item Sort the pair $(U_{j},V_{j})$ for integer $U_{j}.$
    \item Compute integral part $m_{1}=\lfloor
    \frac{U_{j}}{2^{n}}\rfloor.$
    \item Compute integral part $m_{2}=\lfloor
    \frac{V_{j}}{2^{n}}\rfloor.$
    \item Form the integer $m=m_{1}\cdot2^{n}+m_{2}.$
    \item Transform the number $m$ to the message \textbf{M}.
    \item Return the message \textbf{M}.\\
\end{enumerate}
\end{proposition}

\noindent We now proceed to give a proof of correctness.
\newline

Along will begin by computing $W\equiv Cd\equiv V^{2}(\textrm{mod
}pq).$ Along will then have to solve $W\equiv V^{2}(\textrm{mod
}pq)$ using the Chinese Remainder Theorem.\\

\begin{lemma}
Let $p$ and $q$ be two different primes such that
$p\equiv3(\textrm{mod }4)$ and $q\equiv3(\textrm{mod }4)$. Define
$x_{p}$ and $x_{q}$ by $$x_{p}\equiv W^{\frac{p+1}{4}}(\textrm{mod
}p), x_{q}\equiv W^{\frac{q+1}{4}}(\textrm{mod }q).$$ Then the
solutions of the equation $x^{2}\equiv W(\textrm{mod }p)$ are $\pm
x_{p}(\textrm{mod }p)$ and the solutions of the equation
$x^{2}\equiv W(\textrm{mod }q)$ are $\pm x_{q}(\textrm{mod }q).$
\end{lemma}

\begin{lemma}
Let $p$ and $q$ be two different primes such that
$p\equiv3(\textrm{mod }4)$ and $q\equiv3(\textrm{mod }4)$. Define
$x_{p}$ and $x_{q}$ by $$x_{p}\equiv W^{\frac{p+1}{4}}(\textrm{mod
}p), x_{q}\equiv W^{\frac{q+1}{4}}(\textrm{mod }q).$$ Define
$M_{1}\equiv q^{-1}(\textrm{mod }p)$ and $M_{2}\equiv
p^{-1}(\textrm{mod }q).$ Then the solutions of the equation
$V^{2}\equiv W(\textrm{mod }pq)$ are $$V_{1}\equiv
x_{p}M_{1}q+x_{q}M_{2}p\textrm{ }(\textrm{mod }pq),$$
$$V_{2}\equiv x_{p}M_{1}q-x_{q}M_{2}p\textrm{ }(\textrm{mod
}pq),$$ $$V_{3}\equiv -x_{p}M_{1}q+x_{q}M_{2}p\textrm{
}(\textrm{mod }pq),$$ $$V_{4}\equiv
-x_{p}M_{1}q-x_{q}M_{2}p\textrm{ }(\textrm{mod }pq).$$
\end{lemma}

\noindent To solve the equation $V^{2}\equiv W(\textrm{mod }pq)$,
we use the Chinese Remainder Theorem. Consider the equations
$x^{2}_{p}\equiv W (\textrm{mod }p)$ and $x^{2}_{q}\equiv
W(\textrm{mod }q).$ Then the solution of the equation $V^{2}\equiv
W(\textrm{mod }pq)$ are the four solutions of the four systems
$$
\left\{%
\begin{array}{ll}
    V\equiv \pm x_{p}(\textrm{mod }p)\\
    V\equiv \pm x_{q}(\textrm{mod }q)\\
\end{array}%
\right.
$$
Define $M_{1}\equiv q^{-1}(\textrm{mod }p)$ and $M_{2}\equiv
p^{-1}(\textrm{mod }q).$ We will get explicitly $$V_{1}\equiv
x_{p}M_{1}q+x_{q}M_{2}p\textrm{ }(\textrm{mod }pq),$$
$$V_{2}\equiv x_{p}M_{1}q-x_{q}M_{2}p\textrm{ }(\textrm{mod
}pq),$$ $$V_{3}\equiv -x_{p}M_{1}q+x_{q}M_{2}p\textrm{
}(\textrm{mod }pq),$$ $$V_{4}\equiv
-x_{p}M_{1}q-x_{q}M_{2}p\textrm{ }(\textrm{mod }pq).$$

\noindent It can be seen that solving $V^{2}\equiv W(\textrm{mod
}pq)$, we
will get four solutions $V_{i}$ for $i=1,2,3,4$.\\
\newline
We prove below that only one of them leads to the correct
decryption and consequently, there is no decryption failure.

\begin{lemma}
Let $C$ be an integer representing a ciphertext encrypted by the
$AA_{\beta}$ algorithm. The equation $C=Ue_{A1}+V^{2}e_{A2}$ has
only one solution satisfying $V<2^{2n-1}.$
\end{lemma}

\begin{proof}
Suppose for contradiction that there are two couples of solutions
$(U_{1},V_{1})$ and $(U_{2},V_{2})$ of the equation
$C=Ue_{A1}+V^{2}e_{A2}$ with $V_{1}\neq V_{2}$ and $V_{i} <
2^{2n-1}.$ Then
$U_{1}e_{A1}+V_{1}^{2}e_{A2}=U_{2}e_{A1}+V_{2}^{2}e_{A2}.$ Using
$e_{A1}=p^{2}q$, this leads to
$$(U_{2}-U_{1})p^{2}q=(V_{1}+V_{2})(V_{1}-V_{2})e_{A2}.$$
Since gcd$(p^{2}q,e_{A2})=1$, then
$p^{2}q|(V_{1}+V_{2})(V_{1}-V_{2})$ and the prime numbers $p$ and
$q$ satisfy one of the conditions $$p^{2}|(V_{1}\pm V_{2})
\textrm{  or  } \left\{%
\begin{array}{ll}
    pq|(V_{1}\pm V_{2})\\
    \textit{p  } |(V_{1}\mp V_{2})\\
\end{array}%
\right.$$\\
Observe that $p^{2}>2^{2n}$ and $pq>2^{2n}$ while $|V_{1}\pm
V_{2}|< 2\cdot2^{2n-1}=2^{2n}.$ This implies that none of these
conditions is possible. Hence the equation $C=Ue_{A1}+V^{2}e_{A2}$
has only one solution with the parameters of the scheme
\end{proof}

\subsection{Example} Let $n=16$. Along will choose the primes
$p=62683$ and $q=62483$. The public keys will be
\begin{enumerate}
\item $e_{A1}=245505609868187$
\item $e_{A2}=4106878163802480$
\end{enumerate}
The private keys will be
\begin{enumerate}
\item $pq=3916621889$
\item $d=2486483$
\end{enumerate}
Busu's message will contain the following parameters
\begin{enumerate}
\item $m_{1}=544644664056570$
\item $m_{2}=21777$
\end{enumerate}
Busu will also generate the following ephemeral random session
keys
\begin{enumerate}
\item $k_{1}=54433$
\item $k_{2}=33079$
\end{enumerate}
Busu will then generate
\begin{enumerate}
\item $U=35693832703611425953$
\item $V=1427210551$ and
consequently $V^{2}=2036929956885723601$
\end{enumerate}
The ciphertext will be $C=17128459327562266456602243879187691$.
\newline
To decrypt Along will first compute $W=3215349249$. Along will
then obtain the following root values $$V_{1}=318887097,$$
$$V_{2}=2489411338,$$ $$V_{3}=1427210551,$$ and
$$V_{4}=3597734792.$$ Only $U_{3}=\frac{C-V_{3}^{2}e_{A2}}{e_{A1}}$ will produce an integer
value. That is $U_{3}=35693832703611425953$. Finally, $m_{1}$ and
$m_2$ can be obtained. $_\Box$

 \section{Basic Attacks}
\vspace{0.25cm}

\subsection{Congruence attack}

In this subsection we will observe the security of the ciphertext
equation $C=Ue_{A1}+V^{2}e_{A2}$ when it is treated as a
Diophantine equation. We will observe that solving the
corresponding Diophantine equation parametric solution set for the
unknown parameters $U$ and
$V^2$ will result in exponentially many candidates to choose from.\\

\noindent From $C=Ue_{A1}+V^{2}e_{A2}$ and since
gcd$(e_{A1},e_{A2})=1$ we have
$$U\equiv Ce^{-1}_{A1}\equiv a\textrm{ }(\textrm{mod }e_{A2}).$$
Hence $U=a+e_{A2}j$ for some $j\in\mathbb{Z}$. Replacing into $C$
we have
$$C=Ue_{A1}+V^{2}e_{A2}=(a+e_{A2}j)e_{A1}+V^{2}e_{A2}.$$
Then,
$$V^{2}=\frac{C-(a+e_{A2}j)e_{A1}}{e_{A2}}=\frac{C-e_{A1}a}{e_{A2}}-e_{A1}j,$$
where $\frac{C-e_{A1}a}{e_{A2}}=b\in \mathbb{Z}$. It follows that
the equation $C=Ue_{A1}+V^{2}e_{A2}$ has the parametric solutions
$$U=a+e_{A2}\textit{j } \textrm{and } V^{2}=b-e_{A1}j.$$\\

$\bullet$ \underline{Computing with U}\\

\noindent To find $U=a+e_{A2}j$, we should find an integer $j$
such that $2^{4n}<U<2^{4n+1}$. This gives
$$\frac{2^{4n}-a}{e_{A2}}<j<\frac{2^{4n+1}-a}{e_{A2}}.$$ We know that $2^{3n+4}<e_{A2}<2^{3n+6}.$
Then the difference between the upper and the lower bound is
$$\frac{2^{4n+1}-a}{e_{A2}}-\frac{2^{4n}-a}{e_{A2}}=\frac{2^{4n}}{e_{A2}}>\frac{2^{4n}}{2^{3n+6}}=2^{n-6}.$$
Hence the difference is very large and finding the correct $j$ is
infeasible.\\

$\bullet$ \underline{Computing with $V^{2}$}\\

\noindent To find $V^{2}=b-e_{A1}j$, we should find an integer $j$
such that $2^{4n-4}<V<2^{4n-2}$. This gives
$$\frac{2^{4n-4}-b}{-e_{A1}}>j>\frac{2^{4n-2}-b}{-e_{A1}}.$$ We know that $2^{3n}<e_{A1}<2^{3n+3}.$
Then the difference between the upper and the lower bound is
$$\frac{2^{4n-4}-b}{-e_{A1}}-\frac{2^{4n-2}-b}{-e_{A1}}=\frac{3\cdot2^{4n-4}}{e_{A1}}=3\cdot2^{n-7}.$$
Hence the difference is very large and finding the correct $j$ is
infeasible.\\

\subsection{Coppersmith type attack}

\begin{theorem}
Let $N$ be an integer of unknown factorization. Furthermore, let
$f_{N}(x)$ be an univariate, monic polynomial of degree $\delta$.
Then we can find all solutions $x_0$ for the equation
$f_{N}(x)\equiv 0(\textrm{mod }N)$ with
$$|x_{0}|<N^{\frac{1}{\delta}}.$$ in time polynomial in
$(logN,\delta).$
\end{theorem}

\begin{theorem}
Let $N$ be an integer of unknown factorization, which has a
divisor $b>N^{\beta}$. Furthermore let $f_{b}(x)$ be an
univariate, monic polynimial of degree $\delta$. Then we can find
all solutions $x_{0}$ for the equation $f_{b}(x)\equiv
0(\textrm{mod }b)$ with
$$
|x_{0}|\leq \frac{1}{2}N^{\frac{\beta^{2}}{\delta}-\epsilon}
$$
in polynomial time in $(log N,\delta, \frac{1}{\epsilon})$.
\end{theorem}

$\bullet$ \underline{Attacking V}\\

\noindent With reference to Theorem 1. Let $N=e_{A1}=p^{2}q$ and
$d'\equiv e^{-1}(\textrm{mod }N)$. Compute $W\equiv Cd'\equiv
V^{2}(\textrm{mod }N)$. Let $f_{N}(x)\equiv
x^{2}-W\equiv0(\textrm{mod }N)$. Hence, $\delta=2$. Thus the root
$x_{0}=V$ can be recovered if $V<N^{\frac{1}{2}}\approx2^{1.5n}$.
But since $V\approx2^{2n}$, this attack is infeasible.\\

$\bullet$ \underline{Attacking $d$}\\

\noindent With reference to Theorem 2. We begin by observing
$f_{b}(x)=ex-1\equiv0(\textrm{mod }pq)$ where $pq$ in an unknown
factor of $N=e_{A1}=p^{2}q$. Since $pq>N^{\frac{2}{3}}$ we have
$\beta=\frac{2}{3}$. From $f_{b}(x)$ we also have $\delta=1$. By
the Coppersmith theorem, the root $x_{0}=d$ can be found if
$|x_{0}|<N^{\frac{4}{9}}$.But since $d>N^{\frac{4}{9}}$, this
attack is infeasible.

\subsection{Euclidean division attack} From $C=Ue_{A1}+V^{2}e_{A2}$,
the size of each public parameter within $C$ ensures that
Euclidean division attacks does not occur. This can be easily
deduced as
follows:\\
\begin{enumerate}
\item $\lfloor\frac{C}{e_{A1}}\rfloor \neq U$
\item $\lfloor\frac{C}{e_{A2}}\rfloor \neq V^{2}$\\
\end{enumerate}

\section{Analysis on lattice based attack}

The square lattice attack has been an efficient and effective
means of attack upon schemes that are designed based on
Diophantine equations. The $AA_{\beta}$ scheme has gone through
analysis regarding lattice attacks while it went through the
design process. Let $C=Ue_{A1}+V^{2}e_{A2}$ be an $AA_{\beta}$
ciphertext. Consider the diophantine equation
$e_{A1}x_{1}+e_{A2}x_{2}=C$. Introduce the unknown $x_{3}$ and
consider the diophantine equation
$$e_{A1}x_{1}+e_{A2}x_{2}-Cx_{3}=0.$$  Then $(U,V^{2},1)$ is a
solution of the equation. Next let $T$ be a number to fixed later.
Consider the lattice $\mathcal{L}$ spanned by the matrix:
\[
M_{0} = \left( {\begin{array}{*{20}c}
   1 & 0  & {e_{A1}T}  \\
   0 & 1  & {e_{A2}T}  \\
   0 & 0  & {-CT}  \\
\end{array}} \right)
\]
Observe that
$$(x_{1},x_{2},x_{3})M_{0}=(x_{1},x_{2},T(e_{A1}x_{1}+e_{A2}x_{2}-Cx_{3})).$$
This shows that the lattice $\mathcal{L}$ contains the vectors
$(x_{1},x_{2},T(e_{A1}x_{1}+e_{A2}x_{2}-Cx_{3}))$ and more
precisely the vector-solution $V_{0}=(U,V^{2},0).$ Observe that
the length of $V_{0}$ satisfies
$$\|V_{0}\|=\sqrt{U^{2}+V^{4}}\approx 2^{4n}.$$ On the other hand,
the determinant of the lattice is det$(\mathcal{L})=CT$ and the
Gaussian heuristics for the lattice $\mathcal{L}$ asserts that the
length of its shortest non-zero vector is usually approximately
$\sigma(\mathcal{L})$ where
$$\sigma(\mathcal{L})=\sqrt{\frac{dim(\mathcal{L})}{2\pi
e}}\textrm{det}(\mathcal{L})^{\frac{1}{\textrm{dim}(\mathcal{L})}}=\sqrt{\frac{3}{2\pi
e}}(CT)^{\frac{1}{3}}.$$ If we choose $T$ such that
$\sigma(\mathcal{L})>\|V_{0}\|$, then $V_{0}$ can be among the
short non-zero vectors of the lattice $\mathcal{L}$. To this end,
$T$ should satisfy
\begin{equation}
T>(\frac{\pi e}{2})^{\frac{3}{2}} \cdot \frac{2^{12n}}{C}
\end{equation}
Next, if we apply the LLL algorithm to the lattice $\mathcal{L}$,
we will find a basis $(b_{1},b_{2},b_{3})$ such that
$\|b_{1}\|\leq \|b_{2}\| \leq \|b_{3}\|$ and $$b_{i}\leq
2^{\frac{n(n-1)}{4(n+1-i)}}\textrm{det}(\mathcal{L})^{\frac{1}{n+1-i}},
\textrm{for } i=1,...,4 \textrm{ and } n=3.$$ For $i=1$, we choose
$T$ such that
$\|V_{0}\|\leq\|b_{1}\|\leq2^{\frac{1}{2}}(CT)^{\frac{1}{3}}$.
Using the approximation $\|V_{0}\|\approx2^{4n}$, this is
satisfied if
$$V>2^{-\frac{1}{2}}\cdot \frac{2^{12n}}{C},$$ which follows from
the lower bound of equation $(3)$. We experimented this result to
try to find $(U,V^{2},0)$. The LLL algorithm outputs a basis with
a matrix in the form
\[
M_{1} = \left( {\begin{array}{*{20}c}
   a_{11} & a_{12} & 0 \\
   a_{21} & a_{22} & 0 \\
   a_{31} & a_{32} & T \\
   \end{array}} \right)
\]
If $(U,V^{2},0)$ is a short vector, then
$(U,V^{2},0)=(x_{1},x_{2},x_{3})M_{1}$ for some short vector
$(x_{1},x_{2},x_{3})$. We then deduce the system
$$
\left\{%
\begin{array}{ll}
    a_{11}x_{1}+a_{21}x_{2}=U\\
    a_{12}x_{1}+a_{22}x_{2}=V^{2}\\
\end{array}%
\right.
$$
from which we can deduce that $x_{3}=0$. If we compute
$(Ue_{A1}-V^{2}e_{A2})/C$, we get $x_{2}=1$ for some $x_1$. It
follows that
$$
\left\{%
\begin{array}{ll}
    a_{11}x_{1}+a_{21}=U\\
    a_{12}x_{1}+a_{22}=V^{2}\\
\end{array}%
\right.\\
$$
This situation is similar to the congruence attack. We can also
observe that this is a system of two equations with three unknowns
(i.e. $x_{1},U,V$).\\

\subsection{Example with lattice based attack}

We will use the parameters in the earlier example. Observe the
lattice $\mathcal{L}$
spanned by the matrix:\\
\[
M_{0} = \left( {\begin{array}{*{20}c}
   1 & 0 & {e_{A1}T}  \\
   0 & 1 & {e_{A2}T}  \\
   0 & 0 & -CT \\
\end{array}} \right)
\]
\\
\noindent the length of the vector $V=(U,V^{2},0)$ is
approximately $\parallel V
\parallel \approx 35751905917344588937$. We will
use $T=2^{20n}$ which would result in the length of the vector $V$
is shorter than the gaussian heuristic of the lattice $\mathcal{L}$. \\

The LLL algorithm outputs the matrix $M_{1}$ given by:\\
\[
\left( {\begin{array}{*{20}c}
  -4106878163802480 & 245505609868187 & 0  \\
 247367271832221073 & 4155888875658045598 & 0  \\
-1118395942494397 & 66856738131713 & T  \\
\end{array}} \right)
\]
\\

\section{Underlying security principles}

\subsection{The integer factorization problem}

To find the unknown composite $p$ and $q$ such that \noindent
$e_{A1}=p^{2}q$.

\subsection{The square root modulo problem}


Since gcd$(e_{A1},e_{A2})=1$, one can obtain the relation
$V^{2}\equiv \alpha (\textrm{mod } e_{A1})$. Since
$e_{A1}=p^{2}q$, then this is equivalent to calculating square
roots modulo composite integers with unknown factorization which
is infeasible.

\subsection{The modular reduction problem}


Since gcd$(e_{A1},e_{A2})=1$, one can obtain $U\equiv \beta
(\textrm{mod } e_{A2})$. Since $U\gg e_{A2}$, to compute $U$ prior
to modular reduction by $e_{A2}$ is infeasible.

\subsection{Equivalence with integer factorization}
From $C=Ue_{A1}+V^{2}e_{A2}$ we have $$C\equiv V^{2}(\textrm{mod
}e_{A1})$$ where $e_{A1}=p^{2}q$ is of unknown factorization. We
show here that solving this congruence relation is equivalent to
factoring $e_{A1}$. If we know the factorization of $e_{A1}$, then
it is easy to solve the congruence relation. Conversely, suppose
that we know all the solutions. By Lemma 2, the four solutions are
$$V_{1}\equiv
x_{p}M_{1}q+x_{q}M_{2}p\textrm{ }(\textrm{mod }pq),$$
$$V_{2}\equiv x_{p}M_{1}q-x_{q}M_{2}p\textrm{ }(\textrm{mod
}pq),$$ $$V_{3}\equiv -x_{p}M_{1}q+x_{q}M_{2}p\textrm{
}(\textrm{mod }pq),$$ $$V_{4}\equiv
-x_{p}M_{1}q-x_{q}M_{2}p\textrm{ }(\textrm{mod }pq).\\$$ \newline
and are such that $V_{i}<pq$ for $i=1,2,3,4$. We will now have
$V_{1}+V_{3}=2x_{q}M_{2}p+\alpha pq$ for some integer $\alpha$.
Then $V_{1}+V_{3}\equiv 0(\textrm{mod }p)$. On the other hand,
$V_{1}+V_{3}<2pq<p^{2}q$. Hence $V_{1}+V_{3}\not\equiv
0(\textrm{mod }p^{2}q)$. Therefore $$p=
gcd(e_{A1},V_{1}+V_{3})=gcd(p^{2}q,V_{1}+V_{3}).$$ Hence
$q=\frac{p^{2}q}{p^2}$.

\section{Table of Comparison}

The following is a table of comparison between RSA, ECC, NTRU and
$AA_{\beta}$. Let $|E|$ denote public key size. The $AA_{\beta}$
cryptosystem has the ability to encrypt large data sets (i.e.
$4n$-bits of data per transmission). The ratio of $M:|E|$ suggests
better economical value per public key bit being used. \\

\begin{center}
\begin{tabular}{|c|c|c|c|c|}
  \hline
  Algorithm & Encryption & Decryption & Ratio  & Ratio \\
   & Speed & Speed & $M:C$ & $M:|E|$  \\
  \hline
  RSA & $O(nlogn^2)$ & $O(nlogn^2)$ & $1:1$ & $1:2$ \\
  \hline
  ECC & $O(n^3)$ & $O(n^3)$ & $1:2$ & $1:2$ \\
  \hline
  NTRU & $O(n^2)$ & $O(n^2)$ & Varies \cite{hoffstein2} & N/A \\
  \hline
  $AA_{\beta}$ & $O(n^2)$ & $O(n^2)$ & $1:1.75$ & $1:1.5$  \\
  \hline
\end{tabular}
\end{center}

\begin{center}
{Table 1. Comparison table for input block of length $n$}
\end{center}

For the decryption process off the $AA_{\beta}$ scheme, the value of
$k$ stated in its complexity is a ``small" constant. Empirical
evidence for length $n=512$ (i.e. the length of the prime is 512
bits) which would result in total public key length to be $6n=3072$
bits, when decrypting

\section{Conclusion}

The asymmetric scheme presented in this paper provides a secure
avenue for implementors who need to transmit up to $4n$-bits of
data per transmission. With an expansion rate of $1:1.75$ the
ciphertext to be transmitted is not much more larger than the
ratio of the ECC. Eventhough its expansion rate is larger than
RSA, this is only natural since it is transmitting a larger data
set. This will give a significant contribution in a niche area for
implementation of
asymmetric type security in transmitting large data sets.\\
\indent The scheme is also comparable to the Rabin cryptosystem
with the advantage of having a unique decryption result. It has
achieved an encryption and decryption speed with complexity order
of $O(n^2)$ and it also has a simple mathematical structure for
easy implementation.

\section*{Acknowledgments}The authors would like to thank Prof. Dr.
Abderrahmane Nitaj of D\'{e}partement de Math\'{e}matiques,
Universit\'{e} de Caen, France, Dr. Yanbin Pan of Key Laboratory
of Mathematics Mechanization Academy of Mathematics and Systems
Science, Chinese Academy of Sciences Beijing, China and Dr. Gu
Chunsheng of School of Computer Engineering, Jiangsu Teachers
University of Technology, Jiangsu Province, China for valuable
comments and discussion on all prior $AA_{\beta}$ designs.


\bibliographystyle{amsplain}

\end{document}